\makeatletter
\def\input@path{{template/}}
\makeatother

\documentclass[runningheads]{llncs}

\newif\ifpublish
\newif\ifhydrozoan
\newif\ifextended 
\hydrozoanfalse
\extendedtrue
\publishtrue

\usepackage{silence}
\usepackage{type1cm}   
\usepackage{microtype}
\usepackage{amsmath}
\usepackage{amssymb}
\usepackage{algorithm}
\usepackage[noend]{algpseudocode}
\usepackage{graphicx}
\usepackage{tikz}
\usetikzlibrary{arrows.meta}
\usepackage{xcolor}
\usepackage{colortbl}
\usepackage{xspace}
\usepackage{float}
\usepackage{etoolbox}
\usepackage{hyperref}
\usepackage{cleveref}
\usepackage{enumitem}

\setlist{nosep, leftmargin=10pt, labelsep=5pt}

\hypersetup{
    colorlinks=true,
    linkcolor={magenta!80!black},
    citecolor={green!55!black},
    urlcolor={black!55},
    pdfborder={0 0 0},
}
\crefname{section}{Section}{Sections}
\Crefname{section}{Section}{Sections}
\crefname{subsection}{Section}{Sections}
\Crefname{subsection}{Section}{Sections}
\crefname{appendix}{Appendix}{Appendices}
\Crefname{appendix}{Appendix}{Appendices}
\crefname{figure}{Figure}{Figures}
\Crefname{figure}{Figure}{Figures}
\crefname{table}{Table}{Tables}
\Crefname{table}{Table}{Tables}
\crefname{equation}{Equation}{Equations}
\Crefname{equation}{Equation}{Equations}
\crefname{algorithm}{Algorithm}{Algorithms}
\Crefname{algorithm}{Algorithm}{Algorithms}
\crefname{theorem}{Theorem}{Theorems}
\Crefname{theorem}{Theorem}{Theorems}
\crefname{lemma}{Lemma}{Lemmas}
\Crefname{lemma}{Lemma}{Lemmas}
\crefname{proposition}{Proposition}{Propositions}
\Crefname{proposition}{Proposition}{Propositions}
\crefname{corollary}{Corollary}{Corollaries}
\Crefname{corollary}{Corollary}{Corollaries}
\crefname{definition}{Definition}{Definitions}
\Crefname{definition}{Definition}{Definitions}
\crefname{remark}{Remark}{Remarks}
\Crefname{remark}{Remark}{Remarks}
\crefname{observation}{Observation}{Observations}
\Crefname{observation}{Observation}{Observations}
\crefname{myclaim}{Claim}{Claims}
\Crefname{myclaim}{Claim}{Claims}
\crefname{fact}{Fact}{Facts}
\Crefname{fact}{Fact}{Facts}
\makeatletter
\crefname{evalclaim}{claim}{claims}
\Crefname{evalclaim}{Claim}{Claims}
\crefname{ALG@line}{line}{lines}
\Crefname{ALG@line}{Line}{Lines}
\makeatother

\spnewtheorem{observation}{Observation}{\bfseries}{\itshape}
\spnewtheorem{myclaim}{Claim}{\bfseries}{\itshape}
\spnewtheorem{fact}{Fact}{\bfseries}{\itshape}
\newenvironment{proofsketch}{\par\noindent\emph{Proof sketch.}\space}{\hfill\qed\par\medskip}
\newenvironment{restatement}[2]{\par\medskip\noindent\textbf{\Cref{#1} (#2, restated).}\space\itshape}{\par\medskip}
\newenvironment{appstatement}[3]{\ifextended\begin{restatement}{#1}{#3}\def\appstatementend{\end{restatement}}\else\csname #2\endcsname\label[#2]{#1}\def\appstatementend{\csname end#2\endcsname}\fi}{\appstatementend}

\newcommand{\sysname}{{\upshape\textsf{Barnacle}}\xspace}
\newcommand{\fixed}[1]{\textsc{fixed}-#1\xspace}  
\newcommand{\codelink}{
    \ifpublish
        \url{https://github.com/zenodea/hammerheadv2/tree/feat/leader-scheduler} (commit \texttt{bb99146})
    \else
        \url{https://anonymous.4open.science/r/barnacle}
    \fi
}
\newcommand{\leanlink}{
    \ifpublish
        \url{https://github.com/gdanezis/lean-dag} (commit \texttt{24d727f})
    \else
        \url{https://anonymous.4open.science/r/lean-dag}
    \fi
}

\newcommand{\para}[1]{\par\medskip\noindent\textbf{#1.}~}
\newcommand{\ie}{i.e.\@\xspace}
\newcommand{\eg}{e.g.\@\xspace}
\newcounter{evalclaim}
\renewcommand{\theevalclaim}{\textbf{C\arabic{evalclaim}}}
\newcommand{\claimitem}[1]{\item\refstepcounter{evalclaim}\label{#1}\theevalclaim:}

\makeatletter
\patchcmd{\@makecaption}{\small}{\scriptsize}{}{\ClassError{macros}{caption patch failed}{}}
\makeatother
\makeatletter
\patchcmd{\table}{\setlength\belowcaptionskip{10\p@}}{\setlength\belowcaptionskip{4\p@}}{}{\ClassError{macros}{table caption patch failed}{}}
\makeatother

\newcommand{\algsize}{\fontsize{6.5}{7.5}\selectfont}
\newcommand{\appendixalgsize}{\fontsize{6}{7}\selectfont}
\AtBeginEnvironment{algorithm}{\algsize}
\newcommand{\algcommentsize}{\fontsize{5.5}{6.5}\selectfont}
\algrenewcommand{\algorithmiccomment}[1]{\hfill{\algcommentsize\color{black!60}$\triangleright$ #1}}
\algrenewcommand{\alglinenumber}[1]{{\color{black!60}#1:}}
\floatstyle{boxed}
\restylefloat{algorithm}
\makeatletter
\renewcommand\floatc@plain[2]{\setbox\@tempboxa\hbox{\scriptsize{\@fs@cfont #1:} #2}%
    \ifdim\wd\@tempboxa>\hsize {\scriptsize{\@fs@cfont #1:} #2\par}%
    \else\hbox to\hsize{\hfil\box\@tempboxa\hfil}\fi}
\makeatother

\newcommand{\algvar}[1]{\texttt{#1}\xspace}   
\newcommand{\algfunc}[1]{\textsc{#1}\xspace}  
\newcommand{\pleaderoffset}{\algvar{leaderOffset}}
\newcommand{\pwaveoffset}{\algvar{waveOffset}}
\newcommand{\pwavelength}{\algvar{waveLength}}
\newcommand{\pproposersperround}{\algvar{leadersPerRound}}
\newcommand{\pmaxproposersperround}{\algvar{maxLeaders}}
\newcommand{\pinterval}{\algvar{Interval}}
\newcommand{\pcommitthreshold}{\algvar{commitThreshold}}
\newcommand{\pskipthreshold}{\algvar{skipThreshold}}
\newcommand{\plinksize}{\algvar{linkSize}}
\newcommand{\pbackoff}{\algvar{backoff}}
\newcommand{\plastround}{\algvar{lastRound}}
\newcommand{\pthreshold}{\algvar{threshold}}
\newcommand{\ptrycommit}{\algfunc{TryCommit}}
\newcommand{\ptrydecide}{\algfunc{TryDecide}}
\newcommand{\pupdateleaders}{\algfunc{UpdateLeaders}}
\newcommand{\pgetsubdag}{\algfunc{GetSubDag}}
\newcommand{\pcountdirectcommits}{\algfunc{CountDirectCommits}}
\newcommand{\pexpectedcommits}{\algfunc{ExpectedCommits}}
\newcommand{\pdecider}{\algfunc{Decider}}
\newcommand{\pwavenumber}{\algfunc{WaveNumber}}
\newcommand{\pproposeround}{\algfunc{ProposeRound}}
\newcommand{\pcertifyround}{\algfunc{CertifyRound}}
\newcommand{\pvoteround}{\algfunc{VoteRound}}
\newcommand{\pleaderblock}{\algfunc{LeaderBlock}}
\newcommand{\pskippedleader}{\algfunc{SkippedLeader}}
\newcommand{\psupportedleader}{\algfunc{SupportedLeader}}
\newcommand{\ptrydirectdecide}{\algfunc{TryDirectDecide}}
\newcommand{\ptryindirectdecide}{\algfunc{TryIndirectDecide}}
\newcommand{\pgetleader}{\algfunc{GetLeader}}
\newcommand{\piscert}{\algfunc{IsCert}}
\newcommand{\piscertifiedlink}{\algfunc{IsCertifiedLink}}
\newcommand{\pisvote}{\algfunc{IsVote}}
\newcommand{\pislink}{\algfunc{IsLink}}
\newcommand{\pgetdecisionblocks}{\algfunc{GetDecisionBlocks}}
\newcommand{\pundecided}{\bot}
\newcommand{\pdag}{\text{DAG}}
\newcommand{\pgst}{\text{GST}}

\title{\sysname: Adaptive Multi-Leader Scheduling for DAG-Based Consensus}

\ifpublish
    \author{
        George Danezis\inst{3,1} \and
        Zeno de Angeli\inst{1} \and
        Alexandru Ianov Vitanov\inst{2}\thanks{Work done while at University College London.} \and
        Philipp Jovanovic\inst{3,1} \and
        Lefteris Kokoris-Kogias\inst{3} \and
        Alberto Sonnino\inst{3,1} \and
        Pasindu Tennage\inst{4} \and
        Igor Zablotchi\inst{3}
    }
    \institute{Mysten Labs \and University College London (UCL) \and Columbia University \and Digital Asset}
    \authorrunning{Z. de Angeli et al.}
\else
    \author{}
    \institute{}
\fi

\begin{document}

\maketitle
\ifpublish\else
    \pagestyle{plain}
    \thispagestyle{plain}
\fi

\begin{abstract}
    In DAG-based consensus, all validators propose blocks concurrently, and designated leader blocks drive transaction commit. Having multiple leader slots per round cuts queuing latency, yet production deployments run a single leader because of head-of-line blocking: a slow leader stalls the pipeline for at least one leader timeout, and for several waves when its slot must wait for the fallback indirect decision rule. This risk grows with the leader count. We introduce \sysname, an add-on that adapts the leader count at run time. Every interval, it measures on the agreed committed DAG the fraction of slots decided as commit by the direct rule, and drives the leader count with additive increase, multiplicative decrease. The measurement requires no extra messages and no cryptography, and is deterministic. \sysname is generic over DAG protocols; we instantiate it on four protocols spanning the Byzantine ($3f+1$, $5f+1$), crash-only ($2c+1$), and mixed ($5f+3c+1$) fault models, with proven safety and liveness. Results show \sysname matches the best static leader count in every regime: in a healthy network its latency is 6--13\% lower than a single leader's, and under degradation it matches a single leader while remaining 35--56\% below a static high count. We are currently collaborating with the Sui team to integrate \sysname into the Sui blockchain.
\end{abstract}

\section{Introduction}
\label{sec:introduction}

The demand for high-throughput, low-latency blockchains has driven a surge of research on DAG-based Byzantine fault-tolerant consensus~\cite{dag-rider,narwhal,bullshark,shoal++,mysticeti,sailfish++}, which now runs in several production networks~\cite{sui,iota}. In a DAG-based consensus protocol, every validator proposes a block in every round, referencing previous-round blocks to build a directed acyclic graph. Rather than relying on a single proposer per round to sequence transactions, the protocol deterministically designates certain blocks as \emph{leaders}; committing a leader block commits all transactions in its causal history (the DAG vertices reachable by backwards references). For operators of these deployments, end-to-end latency is the primary metric to optimize. Recent protocols explore multi-leader configurations to cut latency further: Shoal~\cite{shoal}, Shoal++~\cite{shoal++}, Mysticeti~\cite{mysticeti}, Sailfish~\cite{sailfish}, Sailfish++~\cite{sailfish++}, Blue Bottle~\cite{blue-bottle}, Orcaella~\cite{orcaella}, and Nemo-Nemo~\cite{nemonemo} all support electing several leaders per round.

The appeal of multi-leaders stems from the need to reduce queuing latency, the delay between a validator receiving a transaction and its inclusion in a leader block. In a DAG, transactions proposed in non-leader blocks commit only once a subsequent leader block references them. The effect is concrete. In Sui's deployment of Mysticeti~\cite{mysticeti}, validators propose a block roughly every 75\,ms and the end-to-end latency is about 400\,ms. A transaction that reaches the validator leading the next round waits, in expectation, half a block interval (about 35\,ms) before it is included in a leader block; a transaction that reaches any non-leading validator additionally waits a full round (75\,ms) for its block to be referenced by the next leader, a roughly 20\% increase in end-to-end latency. With $L$ leaders among $n$ validators in a round, a transaction lands on a leader with probability $L/n$, so the expected queuing latency shrinks as $L$ grows.

Despite this theoretical advantage, to our knowledge, production deployments currently run a single leader per round~\cite{sui-code}. Even protocols that introduce the multi-leader feature remain cautious; the original Mysticeti paper, for instance, recommends deploying with only two leaders or a small constant~\cite{mysticeti}.

The reason for this conservative operational choice is head-of-line blocking, which pits two key latency factors against each other: every additional leader slot cuts queuing latency, but also adds another leader that can stall the pipeline. The mechanics are as follows. Each round consists of one or more designated leader \emph{slots} (positions for leader blocks). Commits proceed in slot order, and a slot is decided across a \emph{wave} (a fixed number of consecutive rounds where references and votes accumulate). Validators advance a round after assembling a quorum of its blocks, but wait up to a fixed leader timeout for the round's leader blocks before proposing without them. A slow leader therefore blocks at the head of the pipeline, with a cost that depends on how far its block traveled in time. In round-level blocking, the leader block reaches nobody in time; every validator waits out the shared leader timeout, and the slot is then decided as skip by the fast, optimistic \emph{direct decision rule}. In commit-level blocking, the leader block reaches only part of the committee; neither the commit nor the skip quorum exists, so the slot stays undecided, and because leaders commit in slot order, every later slot waits until the fallback \emph{indirect decision rule} resolves it, at least a wave later. Either way, the more slots a round has, the more rounds contain a slow leader and the more blocking occurs; and either way the slot is not decided as commit by the direct rule. Any static choice of leader count is wrong in some regime: a high count degrades performance during network instability, while a low count forgoes the latency benefit in healthy conditions (\Cref{fig:hol}).
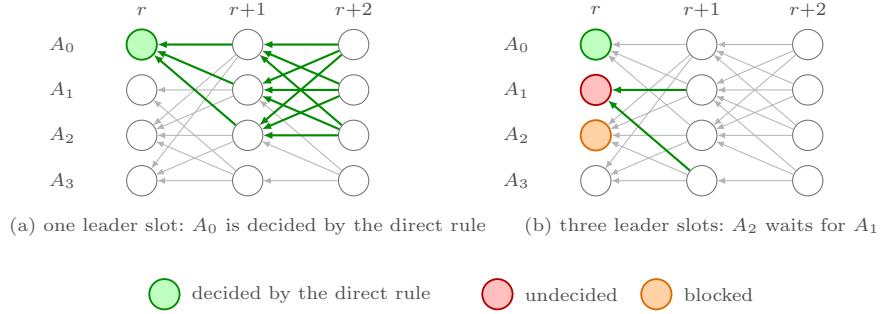
\begin{figure}[t]
    \centering
\tikzset{
    blk/.style={circle, draw=black!55, fill=white, minimum size=0.4cm, inner sep=0pt, line width=0.4pt},
    commit/.style={blk, fill=green!25, draw=green!55!black, line width=0.7pt},
    stall/.style={blk, fill=red!25, draw=red!70!black, line width=0.7pt},
    blocked/.style={blk, fill=orange!35, draw=orange!85!black, line width=0.7pt},
    ref/.style={-{Latex[length=1mm, width=0.8mm]}, black!30, line width=0.3pt},
    vote/.style={-{Latex[length=1.2mm, width=1mm]}, green!55!black, line width=0.8pt},
    lbl/.style={font=\scriptsize, text=black!70},
}
\newcommand{\holdag}[4]{
    \foreach \c/\rl in {0/$r$, 1/$r{+}1$, 2/$r{+}2$} {
        \node[lbl] at (\c, 0.75) {\rl};
    }
    \foreach \a in {0,1,2,3} {
        \node[lbl] at (-0.75, -\a) {$A_\a$};
        \foreach \c in {1,2} { \node[blk] (a\a\c) at (\c, -\a) {}; }
    }
    \node[#1] (a00) at (0, 0) {};
    \node[#2] (a10) at (0, -1) {};
    \node[#3] (a20) at (0, -2) {};
    \node[blk] (a30) at (0, -3) {};
    \foreach \s/\t in {0/0, 0/2, 0/3, 1/0, 1/1, 1/2, 2/0, 2/2, 2/3, 3/1, 3/2, 3/3} {
        \draw[ref] (a\s1) -- (a\t0);
    }
    \foreach \s/\t in {0/0, 0/1, 0/2, 1/0, 1/1, 1/2, 2/0, 2/1, 2/2, 3/1, 3/2, 3/3} {
        \draw[ref] (a\s2) -- (a\t1);
    }
    #4
}
\begin{tikzpicture}[x=1.4cm, y=0.6cm]
    \begin{scope}
        \holdag{commit}{blk}{blk}{
            \foreach \v in {0,1,2} { \draw[vote] (a\v1) -- (a00); }
            \foreach \c in {0,1,2} { \foreach \v in {0,1,2} { \draw[vote] (a\c2) -- (a\v1); } }
        }
        \node[lbl, anchor=north] at (1, -3.6) {(a) one leader slot: $A_0$ is decided by the direct rule};
    \end{scope}
    \begin{scope}[xshift=6.0cm]
        \holdag{commit}{stall}{blocked}{
            \foreach \v in {1,3} { \draw[vote] (a\v1) -- (a10); }
        }
        \node[lbl, anchor=north] at (1, -3.6) {(b) three leader slots: $A_2$ waits for $A_1$};
    \end{scope}
    \begin{scope}[x=1cm, y=1cm, yshift=-3.3cm, xshift=0.3cm]
        \node[commit]  (l1) at (0, 0) {};   \node[lbl, anchor=west, xshift=1pt] at (l1.east) {decided by the direct rule};
        \node[stall]   (l2) at (4.4, 0) {}; \node[lbl, anchor=west, xshift=1pt] at (l2.east) {undecided};
        \node[blocked] (l3) at (6.5, 0) {}; \node[lbl, anchor=west, xshift=1pt] at (l3.east) {blocked};
    \end{scope}
\end{tikzpicture}
    \caption{
        Head-of-line blocking in a DAG. Each column is a round; each row is a validator; nodes are proposed blocks and arrows are references (votes). The same DAG ($n{=}4$; validator $A_1$ has poor network conditions so its round-$r$ block promptly reaches only part of the committee: the commit-level case) is interpreted with one leader slot per round (a) and with three (b). In (a), leader $A_0$'s block gathers enough references (green) and is decided by the direct decision rule at the end of the wave ($r{+}2$). In (b), leader $A_1$'s block gathers too few references for the direct rule and will be decided later by the indirect decision rule; leader $A_2$'s block gathers enough references but, because leaders commit in round-then-slot order, it waits until $A_1$'s slot is decided.
    }
    \label{fig:hol}
\end{figure}

To resolve this tension, we introduce \sysname, a mechanism to dynamically adapt the number of leaders per round to observed network conditions. It re-evaluates the leader count every \pinterval rounds (typically a few seconds).

Measuring network health in a BFT system is itself delicate: validators exchanging latency reports or voting on perceived health would add protocol messages and rest on unverifiable local observations, leaving honest validators without a common value to act on. The fundamental insight behind \sysname is that a DAG-based protocol already exposes the information required to measure network health. In a DAG, every validator proposes a block in every round, so a committed window of the DAG contains every leader block and every vote of that period, and this window is agreed upon by all honest validators. DAG protocols decide each leader with two rules: a fast rule for ideal conditions (the direct decision rule) and a slower fallback that resolves a leader by recursing through later leaders (the indirect decision rule). The recently committed DAG therefore lets each validator count how many leaders of the recent past were decided under ideal conditions, \ie by the direct rule, at no cost. \sysname uses this count as an indicator of network health for the near future, on the assumption that conditions persist for the next few seconds. The measurement requires no extra messages and no cryptography; it is deterministic and computed on agreed state, so every honest validator obtains the same value.

We design \sysname as a generic add-on over a minimal interface exposed by the DAG protocol. The control loop requires only the leader schedule, the committed prefix, and the base protocol's direct-commit predicate. Because this interface is abstract, \sysname is agnostic to the fault model. We instantiate the interface on four base protocols: Mysticeti ($n=3f+1$, Byzantine), Blue Bottle ($n=5f+1$, Byzantine with a two-round commit), Nemo-Nemo ($n=2c+1$, crash-fault tolerant), and Orcaella ($n=5f+3c+1$, mixed: $f$ Byzantine plus $c$ crashed). The mechanism never inspects quorum sizes. While we focus on uncertified DAGs, the same measurement applies to certified DAGs (\Cref{sec:protocol}).

\para{Contributions} We make the following contributions:
\begin{itemize}
    \item We introduce the first mechanism that adapts the leader count of a DAG-based consensus protocol at run time, namely deciding how many validators lead where prior dynamic schedulers~\cite{carousel,hammerhead,shoal} choose \emph{who} leads, with no additional protocol messages or cryptographic tools (\Cref{sec:overview,sec:protocol}).
    \item We prove that safety and liveness hold for any deterministic update rule, by hand (\Cref{sec:proofs}) and machine-checked in Lean~4 (\Cref{app:lean}); the formulation is a generic add-on over an abstract base protocol under explicit assumptions.
    \item We show that one implementation runs unchanged across the $3f+1$, $5f+1$, crash ($2c+1$), and mixed ($5f+3c+1$) fault models, instantiated on Mysticeti, Blue Bottle, Nemo-Nemo, and Orcaella (\Cref{sec:implementation}).
    \item We show that \sysname matches the best static leader count in every regime (\Cref{sec:evaluation}): in a healthy network its latency is $6$--$13\%$ lower than with a single leader per round, and under degradation it matches a single leader while remaining $35$--$56\%$ below a static high count.
\end{itemize}

Furthermore, we are collaborating with the Sui team to integrate \sysname into the Sui blockchain.
\section{Background and Model}\label{sec:preliminaries}

This section fixes the model and the requirements \sysname must meet (\Cref{sec:threat_model}), then recalls the structure of DAG-based consensus protocols that \sysname builds on (\Cref{sec:dag_consensus}).

\subsection{System model and problem}
\label{sec:threat_model}

\para{Fault model}
We consider a system of $n$ validators where a static, computationally bounded adversary controls a subset of them: up to $f$ validators that deviate arbitrarily from the protocol (Byzantine) and up to $c$ that may only halt (crashed).  We consider three fault models; which one applies is fixed by the base protocol \sysname runs on: the Byzantine model admits only Byzantine faults (Mysticeti~\cite{mysticeti} with $n=3f+1$, Blue Bottle~\cite{blue-bottle} with $n=5f+1$); the crash model admits only crashes (Nemo-Nemo~\cite{nemonemo} with $n=2c+1$); and the mixed model admits both (Orcaella~\cite{orcaella} with $n=5f+3c+1$). In all models the remaining validators are honest and follow the protocol. \sysname never inspects $n$, $f$, $c$, or which model applies; it inherits them from the base protocol.

\para{Network}
Validators communicate over a partially synchronous network~\cite{dls88}. Before a Global Stabilization Time ($\pgst$), the adversary controls message timing and delivery, though messages between honest validators are guaranteed to eventually arrive. After $\pgst$, the network becomes synchronous: any message sent at time $x$ is guaranteed to arrive by time $\max\{x, \pgst\} + \Delta$, for a known bound $\Delta$. In this model, after $\pgst$, a base protocol whose leader timeout is at least $\Delta$ decides every honest leader's slot by the direct decision rule (detailed in \Cref{sec:dag_consensus}); only slots led by crashed or Byzantine validators can remain undecided, at most $f{+}c$ per round. This is the model in which the base protocols prove Safety and Liveness (\Cref{def:bab}), and in which we prove that \sysname preserves them.

\para{Deployments}
In practice, validators advance to the next round upon receiving a quorum of blocks, and wait for the round's leader block(s) only for a short, fixed time after that quorum: 200\,ms in Sui's Mysticeti~\cite{sui-code}. Waiting for a pessimistic $\Delta$ every round would result in high latency, so the timeout is tuned to the common case. An honest leader whose block arrives after the timeout, because of congestion or temporary network delays is referenced by too few blocks and its slot is not decided as commit by the direct rule (\Cref{sec:dag_consensus}), as if the leader were faulty. Hence the symptom \sysname acts on, a slot not decided as commit by the direct rule, arises from faulty leaders in the model but from temporary network degradations in deployments; \sysname only ever observes the symptom.

\para{Problem}
\sysname enhances a base protocol that solves Atomic Broadcast.

\begin{definition}[Atomic Broadcast]\label{def:bab}
    Atomic Broadcast satisfies the following properties:
    \begin{itemize}
        \item \textbf{Agreement:} If one honest validator commits a block, all honest validators eventually commit it.
        \item \textbf{Integrity:} Each honest validator commits a block at most once, and only if it was proposed by its author (a requirement that is trivial under crash faults).
        \item \textbf{Validity:} If an honest validator proposes a block, all honest validators eventually commit it.
        \item \textbf{Total Order:} If an honest validator commits a block $b$ before block $b'$, no honest validator commits $b'$ before $b$.
    \end{itemize}
\end{definition}

\sysname must meet two requirements. (R1) Correctness: the base protocol composed with \sysname still satisfies every property of the definition. (R2) Performance, stated informally: (a) in healthy conditions \sysname introduces no latency penalty over the best static leader count; (b) when slots stall, its latency tracks the best static leader count for the prevailing conditions, and it recovers when conditions do. R1 is established formally, by hand-written proofs (\Cref{sec:proofs}) and a Lean formalization (\Cref{app:lean}); R2 is a performance goal outside the formal model and is validated empirically in \Cref{sec:evaluation}.

\subsection{DAG-based consensus}
\label{sec:dag_consensus}

DAG-based protocols operate in rounds. Every validator proposes one block per round, referencing a quorum of previous-round blocks. Validators build a local directed acyclic graph (DAG) from delivered blocks; a block is added only once its whole causal history is present.

\para{Leader slots and decision rules}
Each round has one or more slots in a fixed order. The leader of slot $i$ of round $r$ is a deterministic, public function of $(r, i)$. The leader's round-$r$ block is the leader block of the slot. Every validator proposes every round regardless of how many slots there are; a non-leader's block is an ordinary block.

A wave is a fixed number of consecutive rounds, two or three in the protocols we consider (comprising propose, vote, and optionally certify rounds). The \emph{direct decision rule} decides a slot from the references inside its wave alone: a slot is decided as commit if enough validators reference the leader block in the pattern the protocol prescribes, and as skip if enough validators fail to. Otherwise, it remains undecided, and the \emph{indirect decision rule} later takes over. This rule first deterministically designates a future leader as \emph{anchor}; then decides a slot as commit if the anchor's causal history links to the leader block through enough certified paths, and as skip otherwise. Decided leaders are then ordered in round-then-slot order and their causal histories are linearized deterministically~\cite{dag-rider} into a committed prefix.

\para{Certified and uncertified DAGs}
We describe \sysname over uncertified DAGs, which use best-effort broadcast without explicit certificates: references serve as votes, and a quorum of matching votes in the next round plays the role of a certificate.
\ifextended
    We instantiate \sysname on four base protocols chosen to span fault models: Mysticeti~\cite{mysticeti} ($n=3f+1$, Byzantine, three-round waves), Blue Bottle~\cite{blue-bottle} ($n=5f+1$, Byzantine, two-round waves), Nemo-Nemo~\cite{nemonemo} ($n=2c+1$, crash faults only, two-round waves), and Orcaella~\cite{orcaella} ($n=5f+3c+1$, mixed: $f$ Byzantine plus $c$ crashed, two-round waves). \Cref{tab:base-protocols} summarizes the parameters of their decision rules.
\fi
Nothing in \sysname depends on these values: quorum sizes live exclusively inside the base protocol's decision rules (\Cref{sec:protocol}).
In certified DAGs~\cite{narwhal,bullshark,shoal}, certificates are explicit objects and their direct rule counts references to the leader in the next round; the same measurement applies because every validator's certified block still exists in every round, though this is not evaluated here.

\begin{table}[t]
    \centering
    \footnotesize
    \caption{Base protocols and the parameters of their decision rules: wave length; commit and skip thresholds, the numbers of validators whose references let the direct rule decide a slot as commit or as skip; and link size, the number of certified paths the indirect rule requires.}
    \label{tab:base-protocols}
    \begin{tabular*}{\textwidth}{@{\extracolsep{\fill}}llcccc@{}}
        \hline
        Protocol                       & Fault model             & Wave length & Commit  & Skip                           & Links    \\
        \hline
        Mysticeti~\cite{mysticeti}     & $n{=}3f{+}1$, Byzantine & 3           & $n{-}f$ & $n{-}f$                        & $1$      \\
        Blue Bottle~\cite{blue-bottle} & $n{=}5f{+}1$, Byzantine & 2           & $n{-}f$ & $n{-}f$                        & $n{-}3f$ \\
        Nemo-Nemo~\cite{nemonemo}      & $n{=}2c{+}1$, crash     & 2           & $n{-}c$ & $n$\textsuperscript{$\dagger$} & $1$      \\
        Orcaella~\cite{orcaella}       & $n{=}5f{+}3c{+}1$, mixed & 2          & $n{-}f{-}c$ & $n{-}f{-}c$               & $n{-}3f{-}2c$ \\
        \hline
    \end{tabular*}
    \par\smallskip
    \parbox{\linewidth}{\raggedright $^\dagger$Nemo-Nemo's direct rule never decides a skip; a threshold of $n$ can only trigger when every validator's block is present, so in practice skips are decided by the indirect rule.}
\end{table}

\section{Overview}\label{sec:overview}

\sysname is a control loop on the leader count. The loop runs once every \pinterval rounds (\eg $50$, a few seconds). Its trigger is a \emph{pivot}: the first leader a validator commits more than \pinterval rounds after the previous pivot. On committing a pivot, the validator takes the pivot's causal history over the last \pinterval rounds as the window, counts the slots that the direct decision rule decided as commit, and divides by the number expected under the current leader count: the \emph{direct decision rate}. If the rate is at least \pthreshold (\eg $0.96$) the leader count grows by one, up to \pmaxproposersperround (\eg $5$); otherwise it shrinks by $1, 2, 4, \dots$ on consecutive unhealthy windows, never below one: additive increase, multiplicative decrease, as in TCP congestion control. The new count applies to the rounds after the pivot, with no transition phase. \Cref{alg:abstract} states the rule.

\begin{algorithm}[t]
    \caption{\sysname, abstract form (runs at every party on top of the base protocol)}
    \label{alg:abstract}
    \begin{algorithmic}[1]
        \State $\pinterval$, $\pmaxproposersperround$, $\pthreshold$ \Comment{E.g., $50$ rounds, $5$ leaders, $0.96$}
        \State $\pproposersperround \gets 1$, $\pbackoff \gets 0$, $\plastround \gets 0$ \Comment{State}
        \Statex

        \Procedure{OnCommit}{$b_{leader}$} \Comment{Called after committing $b_{leader}$}
        \State $r \gets b_{leader}.round$
        \If{$r \leq \plastround + \pinterval$} \Return
        \EndIf
        \State $W \gets$ causal history of $b_{leader}$ in rounds $[r - \pinterval, r]$
        \State $observed \gets$ slots of $W$ decided as commit by the direct rule
        \State $expected \gets (\pinterval - \pwavelength + 1) \times \pproposersperround$
        \If{$observed \geq \pthreshold \times expected$} \Comment{Healthy: additive increase}
        \State $\pproposersperround \gets \min(\pproposersperround + 1, \pmaxproposersperround)$
        \State $\pbackoff \gets 0$
        \Else \Comment{Stalled slots: multiplicative decrease}
        \State $\pproposersperround \gets \max(\pproposersperround - 2^{\pbackoff}, 1)$
        \State $\pbackoff \gets \pbackoff + 1$
        \EndIf
        \State $\plastround \gets r$ \Comment{The new count applies to rounds $> r$}
        \EndProcedure
    \end{algorithmic}
\end{algorithm}

\para{Why it works}
Every validator proposes a block in every round, so the window contains every leader block and reference needed. Because the window is the pivot's causal history, all honest validators that commit the pivot hold the same window, compute the same leader count, and agree on every configuration change (\Cref{sec:proofs}). The leader schedule is an interpretation of the DAG, not part of it, so the base protocol's safety applies within each configuration and the only delicate point is the switch, which pinning changes to committed pivot rounds handles. The test is one integer comparison (\eg $100 \times \mathit{observed} \geq 96 \times \mathit{expected}$), with no messages and no cryptography. This agreement argument nowhere uses the shape of the AIMD rule: safety and liveness hold for any deterministic update rule computed on the window (\Cref{sec:proofs}); AIMD is our choice, justified in \Cref{sec:update}.

\para{Why the direct decision rate}
A slot not decided as commit by the direct rule is exactly one that caused head-of-line blocking: either it timed out into a skip, or it stays undecided, waits at least a wave for the indirect decision rule, and holds every later slot with it. The direct decision rate therefore measures blocking in the window, whichever form it took. Measuring under the current leader count only keeps the rule to one comparison per interval, and the multiplicative decrease makes it react within a few intervals when the direct decision rate collapses.

\para{Base protocol interface}
\sysname is defined against any DAG-based consensus protocol that provides the following:
\begin{itemize}
    \item \emph{(A1) Rounds and slots.} The execution is structured in sequential rounds $r = 1, 2, \dots$; in round $r$, every validator proposes one block; slot $i$ of round $r$, for $i$ below the current leader count, is led by a deterministic function of $(r, i)$ and the current configuration.
    \item \emph{(A2) Causal completeness.} A validator adds a newly received block to its local DAG only once the block's entire causal history is available (this claim is formalized in \Cref{sec:proofs}).
    \item \emph{(A3) Direct decision predicate.} There exists a predicate that, given a DAG, decides a slot $(r,i)$ using only blocks of rounds $[r, r + \pwavelength)$, representing the base protocol's direct decision rule alongside the indirect decision rule's certified-link test.  All quorum sizes, and hence the fault model ($n=3f{+}1$ or $5f{+}1$ under Byzantine faults, $2c{+}1$ under crash faults, $5f{+}3c{+}1$ under mixed faults), reside within this predicate; \sysname never inspects them.
    \item \emph{(A4) Base safety and liveness.} For any fixed configuration, honest validators commit the exact same sequence of leaders (safety), and after Global Stabilization Time, they commit new leaders infinitely often (liveness). For the leader rotation used in this paper, the liveness half is proved rather than assumed (\Cref{app:lean}).
\end{itemize}

\section{The \sysname Protocol}\label{sec:protocol}

\Cref{sec:integration} shows how \sysname plugs into the base protocol so that no round ever mixes two leader counts (\Cref{alg:main}) and how the base protocols instantiate the interface; \Cref{sec:update} shows how the window is measured and the count updated so that the result is identical at every honest validator (\Cref{alg:dual_mode}). The base protocol's own decision rules (\Cref{alg:decider,alg:DAG_helper}) are in \Cref{app:base-protocol}.

\subsection{Integration with the base protocol}
\label{sec:integration}

A \emph{configuration} is a leader count together with the round at which it takes effect; the configuration in force at round $r$ is the one with the largest start round at most $r$, and $\ell(r)$ denotes its leader count (\Cref{alg:main}). A validator maintains its current leader count in the variable \pproposersperround.

\begin{algorithm}[t]
    \caption{Protocol Main Function}
    \label{alg:main}
    \begin{algorithmic}[1]
        \State $\pwavelength$ \Comment{Rounds per wave, \Cref{tab:base-protocols}}
        \State $\pproposersperround$ \Comment{Current leader count; $\ell(r)$ is the count in force at round $r$ (\Cref{sec:integration})}
        \State $\pinterval$ \Comment{Rounds between updates, e.g., $50$}
        \State $\plastround \gets 0$ \Comment{Round of the last update}
        \Statex

        \Procedure{\ptrycommit}{$s_{committed}, r_{highest}$} \Comment{$s_{committed}$: the last committed slot}
        \State $S \gets$ \Call{\ptrydecide}{$s_{committed}, r_{highest}$}
        \State $S_{commit} \gets [\;]$ \Comment{Holds committed leaders}
        \For{$s \in S$}
        \If{$s = \pundecided$} \textbf{break}
        \EndIf
        \If{$s = \texttt{Commit}(b_{leader})$}
        \State $S_{commit} \gets S_{commit} \parallel b_{leader}$
        \If{$b_{leader}.round > \plastround + \pinterval$}
        \State $\pproposersperround \gets$ \Call{\pupdateleaders}{$b_{leader}$} \Comment{For the next round}
        \State \Return $S_{commit}$ \Comment{The pivot round's remaining slots keep the old count}
        \EndIf
        \EndIf
        \EndFor
        \State \Return $S_{commit}$
        \EndProcedure
        \Statex

        \Procedure{\ptrydecide}{$s_{committed}, r_{highest}$}
        \State $S \gets [\;]$ \Comment{Holds decisions}
        \For{$r \gets r_{highest}$ \textbf{down to} $s_{committed}.round$}
        \For{$l \gets \ell(r) - 1$ \textbf{down to} $0$}
        \If{$(r, l) \leq (s_{committed}.round, s_{committed}.index)$} \textbf{break}
        \EndIf
        \State $D \gets$ \Call{\pdecider}{$r \bmod \pwavelength,\ l,\ \pdag$} \Comment{Instance of \Cref{alg:decider,alg:DAG_helper}}
        \State $w \gets D.$\Call{\pwavenumber}{$r$}
        \State $s \gets D.$\Call{\ptrydirectdecide}{$w$}
        \If{$s = \pundecided$} $s \gets D.$\Call{\ptryindirectdecide}{$w, S$}
        \EndIf
        \State $S \gets s \parallel S$
        \EndFor
        \EndFor
        \State \Return $S$
        \EndProcedure
    \end{algorithmic}
\end{algorithm}
\begin{algorithm}[t]
    \caption{Update Leader Algorithm}
    \label{alg:dual_mode}
    \begin{algorithmic}[1]
        \State $\pinterval$, $\plastround$ \Comment{Defined in \Cref{alg:main}}
        \State $\pmaxproposersperround$ \Comment{Upper bound on the leader count, e.g., $5$}
        \State $\pthreshold$ \Comment{Direct decision rate considered healthy, e.g., $0.96$}
        \State $\pbackoff \gets 0$ \Comment{Consecutive decreases}
        \Statex

        \Procedure{\pupdateleaders}{$b_{leader}$}
        \State $W \gets$ \Call{\pgetsubdag}{$b_{leader}$}
        \State $observed \gets$ \Call{\pcountdirectcommits}{$W, \pproposersperround$}
        \State $expected \gets$ \Call{\pexpectedcommits}{$\pproposersperround$}
        \If{$observed \geq \pthreshold \times expected$} \Comment{Healthy: additive increase}
        \State $\pproposersperround \gets \min(\pproposersperround + 1, \pmaxproposersperround)$
        \State $\pbackoff \gets 0$
        \Else \Comment{Stalled slots: multiplicative decrease}
        \State $\pproposersperround \gets \max(\pproposersperround - 2^{\pbackoff}, 1)$
        \State $\pbackoff \gets \pbackoff + 1$
        \EndIf
        \State $\plastround \gets b_{leader}.round$
        \State \Return $\pproposersperround$
        \EndProcedure
        \Statex

        \Procedure{\pgetsubdag}{$b_{leader}$}
        \State $r \gets b_{leader}.round$
        \State \Return $\{b \in \pdag[r'] : r - \pinterval \leq r' \leq r \land \Call{\pislink}{b, b_{leader}}\}$
        \EndProcedure
        \Statex

        \Procedure{\pexpectedcommits}{$k$} \Comment{Decidable slots of the window}
        \State \Return $(\pinterval - \pwavelength + 1) \times k$
        \EndProcedure
        \Statex

        \Procedure{\pcountdirectcommits}{$W, k$} \Comment{Slots $\in W$ decided by direct rule, $k$ leaders}
        \State $count \gets 0$
        \For{$r \in \{b.round : b \in W\}$}
        \For{$l \gets 0$ \textbf{to} $k - 1$}
        \State $D \gets$ \Call{\pdecider}{$r \bmod \pwavelength,\ l,\ W$} \Comment{\Cref{alg:decider,alg:DAG_helper} reading $W$}
        \State $w \gets D.$\Call{\pwavenumber}{$r$}
        \For{$b \in D.$\Call{\pleaderblock}{$w$}} \Comment{Leader block(s) of the slot}
        \If{$D.$\Call{\psupportedleader}{$w, b$}} $count \gets count + 1$
        \EndIf
        \EndFor
        \EndFor
        \EndFor
        \State \Return $count$
        \EndProcedure
    \end{algorithmic}
\end{algorithm}

\para{Decision and update procedures}
The \ptrydecide procedure (\Cref{alg:main}) is the base protocol's unchanged mechanism, except that the number of slots it evaluates per round is the leader count of the configuration in force. For each round and each slot, the validator instantiates a \pdecider (\Cref{alg:decider,alg:DAG_helper}, \Cref{app:base-protocol}) with the wave offset, the slot index, and the local DAG \pdag, and applies the direct decision rule followed by the indirect decision rule.

The \ptrycommit procedure (\Cref{alg:main}) walks the decision sequence produced by \ptrydecide in order up to the first undecided slot, appending any committed leaders to the committed prefix. When it commits a leader whose round exceeds $\plastround + \pinterval$, this leader becomes the pivot where the next leader count is decided. The validator calls \pupdateleaders (\Cref{alg:dual_mode}) and returns immediately. The new configuration starts at the round after the pivot's: on the next invocation, the remaining slots of the pivot's round are decided under the old leader count, so every slot of rounds up to and including the pivot's is decided under the configuration in force at its round, and no round ever mixes two counts. Triggering the update when the round is more than \pinterval rounds after the last update, rather than strictly at a multiple of \pinterval, makes the logic robust to the pivot round's slot being skipped.

\para{Base protocols}
\ifextended
The only protocol-specific code is the decider (\Cref{alg:decider}), which is parameterized by \pwavelength, \pcommitthreshold, \pskipthreshold, and \plinksize (\Cref{tab:base-protocols}, \Cref{app:base-protocol}). Mysticeti~\cite{mysticeti} handles Byzantine validators ($n=3f+1$) using three-round waves, thresholds of $n-f$, and one certified link. Blue Bottle~\cite{blue-bottle} handles Byzantine validators ($n=5f+1$) using two-round waves, so the vote block is itself the certificate; it uses thresholds of $n-f = 4f+1$, and $\plinksize = n-3f = 2f+1$ certified links for the indirect decision rule, which is what keeps a commit by the indirect rule consistent with a skip by the direct rule once the vote is the certificate. Nemo-Nemo~\cite{nemonemo} handles crashed validators ($n=2c+1$) with two-round waves, a commit threshold of $c+1$, a skip threshold of $n$ (so the direct rule never decides a skip in practice), and one link. Orcaella~\cite{orcaella} handles $f$ Byzantine plus $c$ crashed validators ($n=5f+3c+1$) with two-round waves, so the vote block is again the certificate; it uses thresholds of $n-f-c = 4f+2c+1$ and $\plinksize = n-3f-2c = 2f+c+1$ certified links, generalizing Blue Bottle's consistency argument to the mixed fault model, and setting $c=0$ recovers Blue Bottle exactly. \sysname's code (\Cref{alg:main,alg:dual_mode}) is identical in all four.
\fi

In certified DAGs, such as Narwhal and Bullshark~\cite{narwhal,bullshark}, or Shoal~\cite{shoal}, certificates are explicit objects and the direct decision rule counts references to the leader in the next round. Every validator's certified block still exists in every round, so the window measurement applies unchanged; we do not evaluate this.

\subsection{Updating the leader count}
\label{sec:update}

Upon committing a pivot, \pupdateleaders (\Cref{alg:dual_mode}) proceeds in three steps: it extracts the window from the pivot's causal history, measures how many of the window's slots the base protocol decided as commit by the direct rule, and updates the leader count from the resulting direct decision rate.

\para{Step 1: the window}
The procedure \pgetsubdag (\Cref{alg:dual_mode}) defines the window $W$ as the pivot's causal history restricted to the last \pinterval rounds. Because it is a function of the pivot only, the window is identical at every honest validator that commits it. It contains every block of those rounds that the pivot references; in particular, every leader block and every reference the direct decision rule reads, because every validator proposes a block in every round.

\para{Step 2: the measurement}
For each round of $W$ and each slot below the current leader count, \pcountdirectcommits (\Cref{alg:dual_mode}) instantiates a \pdecider whose DAG is $W$ rather than the local DAG \pdag. It takes the slot's leader block, or blocks, via \pleaderblock (\Cref{alg:DAG_helper}), and counts those for which \psupportedleader (\Cref{alg:decider}) holds, \ie, the slot is decided as commit by the direct rule within $W$. Evaluating on $W$ rather than the local DAG is what keeps the count identical across validators. The \pexpectedcommits procedure (\Cref{alg:dual_mode}) computes the slots of the window that can be decided within it as $(\pinterval - \pwavelength + 1)$ rounds times the leader count. The window spans $\pinterval + 1$ rounds, of which the last \pwavelength are excluded rather than counted as failures: the $\pwavelength - 1$ rounds below the pivot's have their certify round above the pivot, outside $W$, and the pivot's own round contributes a single block to $W$ (the pivot itself, since no other block of that round can be its ancestor), so its slots can never reach the commit threshold within $W$.

\para{Step 3: the update}
The \pupdateleaders procedure calculates the direct decision rate as the ratio of observed to expected decisions. If this rate is at least \pthreshold, the leader count increases by one, capped at \pmaxproposersperround, and \pbackoff is reset to 0. Otherwise, the leader count decreases by $2^{\pbackoff}$, floored at 1, and \pbackoff is incremented. Additive increase probes for more parallelism cautiously, one leader per interval, so a wrong step costs at most one interval of one stalled slot. Multiplicative decrease sheds parallelism within a few intervals when slots stall, because a single stalled slot already blocks every slot behind it. The cost of the one-interval lag before the decrease is evaluated in \Cref{sec:evaluation}. In the implementation, the threshold test is executed as an integer comparison.

\section{Correctness}
\label{sec:proofs}

\ifextended
    This section establishes that \sysname preserves the correctness of the base protocol, relying solely on the interface A1--A4 from \Cref{sec:overview}; nothing below is specific to a base protocol or to the AIMD rule. We provide the statements and proof sketches; \Cref{app:proofs} contains the full proofs, and every statement is machine-checked in Lean (\Cref{app:lean}). We open source our Lean formalization.\footnote{\leanlink}

    \para{Notation}
    A \emph{configuration} $C$ is a pair consisting of a start round and a leader count $\ell(C)$; the configuration in force at round $r$ is the one with the largest start round at most $r$ (\Cref{sec:integration}). Every validator starts in $C_0 = (0, 1)$. A validator \emph{closes} $C_k$ when it commits the pivot of $C_k$, which is the first committed slot whose round exceeds the start round of $C_k$ by more than \pinterval. The pivot's round becomes the start round of $C_{k+1}$. Because a validator's DAG is finite at any given time, its run closes only finitely many configurations. Accordingly, every statement below compares validators based on the configurations that \emph{both} have reached, and liveness asserts that runs reach every height, never that a completed infinite sequence exists. After Global Stabilization Time (GST), messages between honest validators arrive within a known bound $\Delta$.

    \begin{fact}
        \label{fact:dag}
        Let $b$ be a block held by an honest validator. (i)~The validator's DAG contains the entire causal history of $b$. (ii)~Every honest validator eventually holds $b$. (iii)~Any two honest validators that hold $b$ hold the same causal history of $b$.
    \end{fact}

    \begin{proofsketch}
        (i) is A2; (ii) is the delivery guarantee of the broadcast layer; (iii) follows by induction on the causal history, using (i) and (ii).
    \end{proofsketch}

    \begin{lemma}
        \label[lemma]{lemma:window}
        Any two honest validators that commit the same pivot compute the same window.
    \end{lemma}

    \begin{proofsketch}
        The window is a function of the pivot's causal history and \pinterval alone, and the histories coincide by \Cref{fact:dag}(iii); \pgetsubdag is deterministic.
    \end{proofsketch}

    \begin{proposition}
        \label[proposition]{prop:agreement}
        For any deterministic update rule, such as the AIMD rule of \Cref{alg:dual_mode}, any two honest validators agree on every configuration both have reached, its start round and its leader count, and on the verdict of every slot of every configuration both have closed.
    \end{proposition}

    \begin{proofsketch}
        Induction over configurations. Both validators start in $C_0$. Within $C_k$, every verdict up to and including the pivot's round is a derivation against the one fixed schedule of $C_k$ (\Cref{sec:integration}), so the two validators' verdicts --- in particular the pivot itself --- are equal by the safety half of A4. The next configuration is one deterministic function of one agreed pivot and its agreed window (\Cref{lemma:window}), so $C_{k+1}$ is equal at both. A validator that has not reached the pivot has closed a strict prefix and agrees on it, and closes the gap once its view converges (\Cref{fact:dag}(ii)). No synchrony assumption and no base-protocol fact beyond A4-safety is used.
    \end{proofsketch}

    \begin{theorem}
        \label[theorem]{thm:safety}
        \sysname composed with any base protocol satisfying A1--A4 satisfies the safety of \Cref{def:bab}: the committed sequences of any two honest validators agree as far as both reach and grow only by appending (Agreement, Total Order), and no block is committed twice (Integrity).
    \end{theorem}

    \begin{proofsketch}
        A validator's committed sequence is determined by its verdicts, which agree configuration by configuration (\Cref{prop:agreement}); hence the sequences agree on the common range and extend by prefixes. For Integrity, the commit boundary is a slot, and a commit verdict names a block of its slot's round and leader, so a block occupies exactly one slot.
    \end{proofsketch}

    \begin{remark}
        \label[remark]{remark:conservativity}
        With the constant update rule, \sysname \emph{is} the base protocol: the same schedule, the same verdicts, the same ledger. The AIMD rule additionally keeps the leader count within $[1, \pmaxproposersperround]$, and its threshold test is one integer comparison.
    \end{remark}

    \begin{lemma}
        \label[lemma]{lemma:progress}
        After GST, if every honest validator has closed configuration $C_k$, then some honest validator closes $C_{k+1}$ within bounded time, and every honest validator does within a further $\Delta$.
    \end{lemma}

    \begin{proofsketch}
        The liveness half of A4 supplies committed slots past the threshold round of $C_{k+1}$'s pivot; the least such slot is the pivot. The lemma consumes base \emph{safety} as well: the commit that liveness supplies and the verdict that the schedule update reads are identified by agreement. Delivery within $\Delta$ then propagates the pivot's history to every honest validator (\Cref{fact:dag}(ii)), and determinism closes $C_{k+1}$ everywhere.
    \end{proofsketch}

    \begin{theorem}
        \label[theorem]{thm:liveness}
        After GST, every honest validator's run reaches every height: for every $k$ it closes configuration $C_k$, and committed slots keep appearing.
    \end{theorem}

    \begin{proofsketch}
        Within a configuration, by the per-configuration clause of A4-liveness; across configurations, by \Cref{lemma:progress} applied repeatedly. The leader count may oscillate, but configuration start rounds strictly increase. In the finite-horizon form of the statement (\Cref{app:lean}), every slot lying at least a wave plus the commit gap below the horizon is decided; the margin is necessary, not an artifact: no protocol can decide a slot at the last round under the horizon, since no decision anchor's wave fits above it.
    \end{proofsketch}

    \begin{remark}
        \label[remark]{remark:rotation}
        For the leader schedule used in this paper, the per-configuration liveness clause of A4 is proved rather than assumed, for all four base protocols at every leader count (\Cref{app:lean}). Because the classic argument --- a wave of consecutive honest-led slots --- has no instance once several leaders share a round, a new argument replaces it: a descent through the first slot of each round and a pigeonhole on the rotation, yielding a commit gap of $n + \pwavelength - 1$ rounds.
    \end{remark}

\fi

\para{What \sysname does not guarantee}
The chosen leader count is a heuristic, not proved optimal. Byzantine validators can bias one window's measurement and make the count suboptimal for one interval, but the count stays within $[1, \pmaxproposersperround]$ and never falls below the single-leader baseline (\Cref{sec:evaluation}).

\section{Implementation}\label{sec:implementation}
\ifextended
    \sysname is a scheduler module of ${\sim}270$ LOC that adds no messages, no cryptography, and no persistent state beyond the current configuration, and runs unchanged on all four base protocols. We implement it in Rust by forking the Mysticeti codebase~\cite{mysticeti-code}, inheriting its networking (\texttt{tokio}~\cite{tokio} over TCP), signatures (\texttt{ed25519-consensus}~\cite{ed25519-consensus}), and write-ahead-log crash recovery. Our simulator replaces the networking stack behind the same interface: a discrete-event controller emulates the \texttt{tokio} runtime and the TCP links, so every other component, consensus logic included, runs unmodified in simulation and in deployment. We are open-sourcing our implementation, along with the simulator and scripts\footnote{\codelink}.

    The scheduler module implements \Cref{alg:dual_mode}: it extracts the window, counts the slots decided as commit by the direct rule within the window using the base protocol's own wave length and commit threshold, and applies the AIMD update. An additional ${\sim}500$ LOC integrate the scheduler into the commit path and the protocol configuration. In the commit path (\Cref{alg:main}), the committer truncates the decision sequence at the pivot, invokes the scheduler, and rebuilds its decider instances for the new leader count. The protocol configuration gains three parameters: an adaptive scheduling toggle, the maximum leader count, and the interval. This code runs unchanged across all protocol variants, as the scheduler reads the wave length and thresholds from the protocol description and never inspects the fault model.

    \para{Simulation layer}
    We inherit this simulation layer from the base codebase, and extend it (${\sim}600$ LOC) to model time-varying, heterogeneous networks, where every link has an independent one-way latency range in each direction. A run is scripted as a sequence of phases. Each phase can designate specific validators as slow, degrading every link that touches them, or optionally only their outbound links to a chosen subset of validators, beyond the leader timeout. In the configuration of \Cref{sec:evaluation} a slow validator's blocks reach every other validator late, so its leader slots miss the direct rule after the leader timeout; degrading only part of its outbound links instead reproduces the partial-quorum situation of \Cref{fig:hol}. The simulator records, per validator and sampling window, the committed latency, the leader count in force, and the leader-slot decisions by type, to support the plots and numbers of \Cref{sec:evaluation}.
\fi

\section{Evaluation}\label{sec:evaluation}
This section benchmarks \sysname on the base protocols of \Cref{sec:implementation}. Our evaluation has a single goal: quantify what adapting the leader count buys over the best \emph{static} leader count in each regime. Our evaluation makes the following claims:
\begin{itemize}
    \claimitem{claim:c1} In a healthy network, \sysname matches the latency of the best static leader count.
    \claimitem{claim:c2} When slow leaders stop slots from being decided as commit by the direct rule, \sysname reduces the leader count until that stops happening and matches the latency of the best static leader count for the degraded regime, well below that of a high static count.
    \claimitem{claim:c3} \sysname quickly adapts in both directions: it reduces the leader count when the network degrades and increases it when it heals; the transition adds little to the degraded-phase latency.
    \claimitem{claim:c4} \Cref{claim:c1,claim:c2,claim:c3} hold unchanged across base protocols with different fault models, wave lengths, and decision quorums, with identical scheduler code and parameters.
\end{itemize}
Benchmarking and testing BFT protocols under actual Byzantine behavior is an open problem~\cite{twins}; the state of the art establishes worst-case guarantees through formal proofs, which we give in \Cref{sec:proofs} for both safety and liveness.

\subsection{Experimental setup}
\ifextended
    \para{Simulator and scenario}
    We rely on simulation because degradations must be identical across base protocols and leader configurations, switch on and off at scripted times, and remain reproducible. None of this is possible in a cloud deployment, where slowness cannot be injected on demand and never repeats identically. We simulate a committee of $n{=}10$ validators in a full mesh where the one-way link latency is uniform in $[50, 100]$\,ms. The execution consists of three phases: a healthy phase of $100$\,s, a degraded phase of $400$\,s, and a healthy recovery phase of $100$\,s. During the degraded phase, a set of \emph{slow} validators has every link that touches it, in both directions, raised to a latency uniform in $[1000, 1300]$\,ms, beyond the $1$\,s leader timeout; all other links are unchanged. The slow set is as large as the base protocol tolerates while the remaining validators still form a quorum: three validators for Mysticeti ($f{=}3$)\ifhydrozoan, Nemo-Nemo ($c{=}3$), and Hydrozoan\else{} and Nemo-Nemo ($c{=}3$)\fi, one for Blue Bottle, whose $4n/5{+}1$ quorum tolerates a single fault at $n{=}10$, and two for Orcaella ($f{=}1$, $c{=}1$), whose $n{-}f{-}c$ quorum of eight tolerates exactly two faults at $n{=}10$.  A slow validator falls rounds behind, so when it is a leader its block reaches nobody in time: every other validator waits out the leader timeout before proposing without it, and the slot is then skipped: directly for Mysticeti, Blue Bottle, and Orcaella (round-level blocking), or through the indirect decision rule for Nemo-Nemo\ifhydrozoan{} and Hydrozoan\fi, whose direct-skip quorum is unanimity, in which case the slot holds every later slot until an anchor resolves it (commit-level blocking, \Cref{fig:hol}). Either way the slot is not decided as commit by the direct rule and lowers the direct decision rate, and with more leader slots per round, more rounds contain a slow leader: three rounds in ten with one leader, seven in ten with five. The \sysname scheduler operates with $\pinterval = 25$ rounds (about $2$\,s while healthy and $10$--$20$\,s while degraded, since rounds slow down), $\pmaxproposersperround = 5$, and $\pthreshold = 0.96$.

\para{Configurations and metrics}
We compare three configurations of each base protocol: \sysname, with the scheduler enabled, and two static baselines that never change their leader count, \fixed{1} with a single leader per round and \fixed{5} with five, the two extremes of the range the scheduler may choose from. \sysname runs start with a leader count of five. Our primary metric is end-to-end latency from client submission to commit, measured per validator over sampling windows of $3$\,s while healthy and $5$\,s while degraded. We report the mean over validators, from one run per configuration with a common seed so that the three configurations of a protocol see the same network draws. The per-phase means are summarized in \Cref{tab:eval}; the ``steady'' degraded window is defined as $[200, 500]$\,s, after every \sysname run has reached one leader.
\fi

\begin{table}[t]
    \centering
    \footnotesize
    \caption{Mean end-to-end latency (ms) per phase; the degraded phase is its steady window ($200$--$500$\,s). \sysname rows give the difference to \fixed{1} and to \fixed{5}.}
    \label{tab:eval}
    \begin{tabular*}{\textwidth}{@{\extracolsep{\fill}}llrrr@{}}
        \hline
        Protocol & Configuration & Healthy & Degraded (steady) & Recovery \\
        \hline
        Mysticeti & \fixed{1} & 348 & 2,084 & 424 \\
        & \fixed{5} & 326 & 3,183 & 478 \\
        \rowcolor{black!10} & \sysname & 328 (-6\%, +1\%) & 2,077 (-0\%, -35\%) & 408 (-4\%, -15\%) \\
        \hline
        Blue Bottle & \fixed{1} & 274 & 753 & 301 \\
        & \fixed{5} & 236 & 1,693 & 305 \\
        \rowcolor{black!10} & \sysname & 238 (-13\%, +1\%) & 745 (-1\%, -56\%) & 270 (-10\%, -11\%) \\
        \hline
        Nemo-Nemo & \fixed{1} & 259 & 1,964 & 344 \\
        & \fixed{5} & 232 & 3,484 & 396 \\
        \rowcolor{black!10} & \sysname & 234 (-10\%, +1\%) & 1,965 (+0\%, -44\%) & 308 (-10\%, -22\%) \\
        \hline
        Orcaella & \fixed{1} & 268 & 1,212 & 311 \\
        & \fixed{5} & 233 & 2,130 & 331 \\
        \rowcolor{black!10} & \sysname & 235 (-12\%, +1\%) & 1,209 (-0\%, -43\%) & 284 (-9\%, -14\%) \\
        \hline
    \end{tabular*}
\end{table}

\subsection{Healthy network}
In the first and last $100$\,s of each panel in \Cref{fig:latency}, the three configurations are indistinguishable except that \fixed{1} is slower.  The \sysname configuration equals \fixed{5} within $1\%$ and beats \fixed{1} by $6\%$ (Mysticeti, $328$ vs $348$\,ms), $13\%$ (Blue Bottle, $238$ vs $274$\,ms), \ifhydrozoan\else $10\%$ (Nemo-Nemo, $234$ vs $259$\,ms), and $12\%$ (Orcaella, $235$ vs $268$\,ms).\fi{} The gap to \fixed{1} is the queuing-latency saving of running five leaders instead of one (\Cref{sec:introduction}). \sysname retains this saving in full: every healthy window has a direct decision rate of $1$, so the scheduler never lowers the leader count below five. This confirms \Cref{claim:c1}.

\begin{figure}[t]
    \centering
    \ifhydrozoan
\else
        \includegraphics[width=\linewidth]{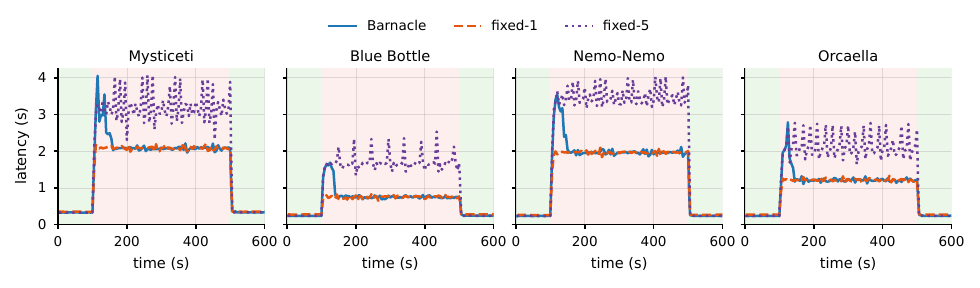}
        \caption{End-to-end latency over time under \sysname and static leader counts (one and five), for Mysticeti, Blue Bottle, Nemo-Nemo, and Orcaella. The network goes healthy $\to$ degraded $\to$ healthy.}
    \fi
    \label{fig:latency}
\end{figure}

\subsection{Degraded network and recovery}
At $t{=}100$\,s latency rises for every configuration, and by an amount that follows the leader count: a round containing a slow leader costs every validator the leader timeout, so the mean round stretches from ${\sim}90$\,ms to ${\sim}380$\,ms under \fixed{1} and ${\sim}730$\,ms under \fixed{5}, and commit latency grows with it. As shown in \Cref{fig:leaders} (Mysticeti), \sysname starts backing off within $15$\,s of the onset, about one interval at the slower round rate, and then follows the $-1, -2, -4$ sequence of the multiplicative decrease, stepping $5 \to 4 \to 2 \to 1$ and reaching one leader about $45$\,s into the phase, after which it remains flat.  Latency (\Cref{fig:latency}, \Cref{tab:eval}) follows the leader count: in the steady window, the Mysticeti \sysname configuration achieves $2.08$\,s, equal to \fixed{1} and $35\%$ below the $3.18$\,s of \fixed{5}; Blue Bottle reaches $0.75$\,s vs $0.75$\,s ($-1\%$) for \fixed{1} and $1.69$\,s ($-56\%$) for \fixed{5}; \ifhydrozoan\else Nemo-Nemo $1.97$\,s vs $1.96$\,s ($0\%$) and $3.48$\,s ($-44\%$); and Orcaella $1.21$\,s vs $1.21$\,s ($0\%$) and $2.13$\,s ($-43\%$).\fi{} Over the whole degraded phase, descent included, the \sysname configuration is $5\%$ (Mysticeti), $9\%$ (Blue Bottle)\ifhydrozoan\else, $6\%$ (Nemo-Nemo), and $6\%$ (Orcaella)\fi{} above \fixed{1}, which bounds the cost of the transient. No configuration stalls: every $5$\,s window of every validator commits transactions throughout the degraded phase. During recovery at $t{=}500$\,s, latency returns to the healthy level for all configurations within seconds. The leader count climbs $1 \to 2 \to 3 \to 4 \to 5$ one step per interval, arriving back at five about $15$\,s after healing.  The recovery numbers in \Cref{tab:eval} are means over the $100$\,s after healing and thus cover two parts: the drain of the backlog, shared by all configurations, and a healthy remainder in which \sysname, back at five leaders, regains the multi-leader advantage while \fixed{1} keeps its single leader. This is why \sysname leads \fixed{1} in this column too. This confirms \Cref{claim:c2} and \Cref{claim:c3}.

\begin{figure}[t]
    \centering
    \includegraphics[width=\linewidth]{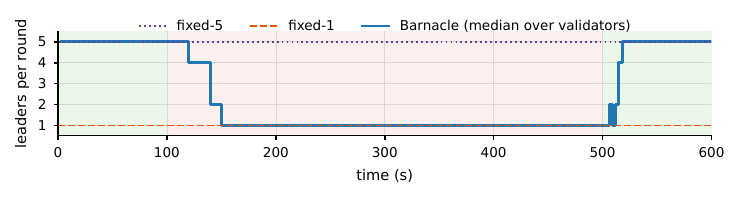}
    \caption{Leader count chosen by \sysname over time (Mysticeti): back-off within one interval of the onset, multiplicative decrease $5 \to 4 \to 2 \to 1$, additive recovery.}
    \label{fig:leaders}
\end{figure}

\subsection{Generality across fault models}
The behavior shown in \Cref{fig:latency,fig:leaders} remains the same across Mysticeti~\cite{mysticeti,mysticeti-code}, Blue Bottle~\cite{blue-bottle}\ifhydrozoan\else, Nemo-Nemo~\cite{nemonemo}, and Orcaella~\cite{orcaella}\fi, employing identical scheduler code and parameters. Only the absolute latencies differ, set by each base protocol's wave length and quorums:  the price of a slow leader slot ranges from a directly skipped slot (Mysticeti, Blue Bottle, Orcaella) to an indirectly decided one (Nemo-Nemo), and the penalty of \fixed{5} over \fixed{1} in the steady window ranges from $53\%$ (Mysticeti) to \ifhydrozoan\else $125\%$ (Blue Bottle)\fi. This confirms \Cref{claim:c4}.  We do not evaluate the sensitivity to \pinterval and \pthreshold; the threshold of $0.96$ tolerates about four non-direct slots per window under five leaders and none under a single leader, and no per-protocol tuning was needed.
\section{Related Work}\label{sec:related_work}

All prior multi-leader DAG protocols fix the leader count at configuration time; to our knowledge, \sysname is the first to adapt it at run time, and the first to compose with any of them as an add-on.

\para{Multi-leader DAG protocols (static leader count)}
The single-leader DAG lineage includes DAG-Rider~\cite{dag-rider}, Narwhal and Tusk~\cite{narwhal}, Bullshark~\cite{bullshark}, and Cordial Miners~\cite{cordial-miners}. To cut latency, recent designs support several leaders per round. Shoal~\cite{shoal} and Shoal++~\cite{shoal++} commit several anchors per round on certified DAGs. Sailfish and Sailfish++~\cite{sailfish,sailfish++} guarantee a leader every round and specify a multi-leader variant. Mysticeti~\cite{mysticeti} operates on uncertified DAGs and recommends two leaders. Blue Bottle~\cite{blue-bottle} targets $n=5f{+}1$ with two-round waves, Nemo-Nemo~\cite{nemonemo} brings the design to the crash-fault model, Orcaella~\cite{orcaella} to a mixed model tolerating $f$ Byzantine and $c$ crashed validators ($n=5f{+}3c{+}1$), and Mahi-Mahi~\cite{mahi-mahi} is an asynchronous protocol with a configurable number of leaders per round. We instantiate \sysname on four of these protocols.

\para{Leader selection: who leads}
Reputation-based election protocols determine who leads. Carousel~\cite{carousel} applies this to chained protocols. Hammerhead~\cite{hammerhead} implements reputation for DAG protocols and is deployed in Sui. Shoal~\cite{shoal} adds reputation to certified DAGs, and Liu et al.~\cite{swle} formalize reputation-based election as a protocol-independent abstraction with generic correctness guarantees, as we do for the leader count. These belong to the same technique family as ours: they derive the schedule deterministically from the committed history, requiring no messages. However, they use a different knob: they exclude unreliable validators from the leader set, whereas \sysname sizes the leader set. The approaches are orthogonal and composable: reputation removes unreliable validators from the leader set, and \sysname then sizes it.

\para{Multi-BFT and multi-proposer protocols}
Parallel consensus systems such as Mir-BFT~\cite{mir-bft}, ISS~\cite{iss}, RCC~\cite{rcc}, and Autobahn~\cite{autobahn} multiplex parallel consensus instances or proposers into one log. The straggler problem of the global ordering, where a slow instance delays the log, is head-of-line blocking in another guise. Ladon~\cite{ladon} attacks this by ordering blocks dynamically. \sysname instead adapts the degree of parallelism; applying it to these protocols is future work.

\para{Adaptive BFT}
BFTBrain~\cite{bftbrain} switches between whole protocols at run time with reinforcement learning, using features that include the fraction of requests committed on the fast path, and coordinates learning through consensus. \sysname adapts a single parameter with a closed-form deterministic rule that needs no learning and no coordination beyond the committed DAG. The rule itself is additive increase, multiplicative decrease, the congestion-control law analyzed by Chiu and Jain~\cite{aimd}, which \sysname applies to regulate the leader count dynamically.

\ifextended
    \section{Conclusion}
\label{sec:conclusion}

Because of head-of-line blocking, no static leader count is right in every regime. We introduced \sysname, which adapts the leader count with an AIMD control loop on the direct decision rate measured on the committed DAG. The mechanism is local, deterministic, and message-free, and preserves the base protocol's safety and liveness. Our evaluation (\Cref{sec:evaluation}) demonstrates that \sysname matches the best static leader count in every regime, reducing latency over a single leader when healthy and avoiding the severe degradation of high leader counts during instability. \sysname is currently being integrated into the Sui blockchain. Because correctness holds for any deterministic update rule, the rule can be swapped without touching the proofs (\Cref{sec:proofs}); future work includes replaying the window under counterfactual leader counts to jump directly to the best one, and extending the approach to multi-proposer protocols outside DAGs.

\fi
\ifpublish
    \section*{Acknowledgements}

This work is partially funded by Mysten Labs.
We thank George Danezis for insightful discussions and for reviewing the Lean~4 formalization of \sysname (\Cref{app:lean}), and Mingwei Tian and Arun Koshy for their feedback on the early design of \sysname.

\fi

\bibliographystyle{splncs04}
\bibliography{references}

\appendix
\section{Base-Protocol Pseudocode}\label{app:base-protocol}
This appendix specifies the base-protocol decision rules that the algorithms of \Cref{sec:protocol} invoke. A single generic decider covers all four base protocols: \Cref{alg:decider} contains the decision rules, parameterized by the three thresholds \pcommitthreshold, \pskipthreshold, and \plinksize; \Cref{alg:DAG_helper} contains the decider instance (wave geometry and DAG predicates) and is identical for every protocol; \Cref{tab:base-protocols} gives the per-protocol values. The pseudocode follows Mysticeti's published algorithms~\cite{mysticeti}, generalized where the four protocols differ.

\subsection{One decider for four protocols}
Six changes over the pseudocode published in the original papers~\cite{mysticeti,blue-bottle,nemonemo,orcaella} make it generic: (i)~The vote round is the round after the propose round, since no protocol decides in fewer than two rounds; for two-round waves the vote and certify rounds coincide. (ii)~Quorum sizes are parameters rather than hard-coded $2f{+}1$. (iii)~\pisvote is a parent lookup: a round-$r$ leader block can only be a direct parent of a round-$(r{+}1)$ block, and taking the first matching parent makes a block vote for at most one block of an equivocating leader. (iv)~\piscert returns \pisvote when the vote and certify rounds coincide: the vote block is itself the certificate. (v)~The skip test is per slot, counting blocks that vote for no block of the slot, not per equivocating leader block. (vi)~\piscertifiedlink counts certified links and compares against \plinksize.

\subsection{The decider instance}
An instance (\Cref{alg:DAG_helper}) is parameterized by a wave offset, a slot index, and the DAG it reads (the local DAG in \Cref{alg:main}, the window in \Cref{alg:dual_mode}). Pipelining uses one instance per offset, making every round the propose round of some wave, a strategy that follows Mysticeti~\cite{mysticeti}. \pleaderblock returns the slot's leader block, or several under equivocation. \pisvote and \pislink are the two primitive relations: direct parenthood of the leader block, and reachability via references.

\begin{algorithm}[t]
    \caption{Decider Instance: Rounds and DAG Helpers (identical for all base protocols)}
    \label{alg:DAG_helper}
    \appendixalgsize
    \begin{algorithmic}[1]
        \State $\pwaveoffset = i$ \Comment{First Decider parameter ($i$): pipelining offset}
        \State $\pleaderoffset = l$ \Comment{Second Decider parameter ($l$): slot in the round}
        \State $\pdag$ \Comment{Third Decider parameter: the DAG the rules read}
        \State $\pwavelength$ \Comment{$3$ (Mysticeti); $2$ (Blue Bottle, Nemo-Nemo)}
        \Statex

        \Procedure{\pwavenumber}{$r$}
        \State \Return $(r - \pwaveoffset) / \pwavelength$
        \EndProcedure
        \Statex

        \Procedure{\pproposeround}{$w$}
        \State \Return $(w \times \pwavelength) + \pwaveoffset$
        \EndProcedure
        \Statex

        \Procedure{\pvoteround}{$w$}
        \State \Return \Call{\pproposeround}{$w$} $+\;1$ \Comment{Two rounds is the minimum}
        \EndProcedure
        \Statex

        \Procedure{\pcertifyround}{$w$}
        \State \Return \Call{\pproposeround}{$w$} $+ (\pwavelength - 1)$ \Comment{$=$ \Call{\pvoteround}{$w$} if $\pwavelength = 2$}
        \EndProcedure
        \Statex

        \Procedure{\pgetdecisionblocks}{$w$}
        \State \Return $\pdag[\Call{\pcertifyround}{w}]$
        \EndProcedure
        \Statex

        \Procedure{\pleaderblock}{$w$} \Comment{Several blocks upon equivocation}
        \State $r_{propose} \gets$ \Call{\pproposeround}{$w$}
        \State $l \gets$ \Call{\pgetleader}{$r_{propose} + \pleaderoffset$} \Comment{Deterministic}
        \State \Return $\{b \in \pdag[r_{propose}] : b.author = l\}$
        \EndProcedure
        \Statex

        \Procedure{\pisvote}{$b_{vote}, b_{leader}$} \Comment{$b_{vote}$ is in round $b_{leader}.round + 1$}
        \State $P \gets [\,b \in b_{vote}.parents : (b.author, b.round) = (b_{leader}.author, b_{leader}.round)\,]$
        \State \Return $P \neq [\;] \land P[0] = b_{leader}$ \Comment{One vote per equivocating leader}
        \EndProcedure
        \Statex

        \Procedure{\pislink}{$b_{old}, b_{new}$}
        \State \Return $\exists$ a sequence of $k \in \mathbb{N}$ blocks $b_1, \ldots, b_k$ s.t.
        \Statex \hspace{2em} $b_1 = b_{old} \land b_k = b_{new} \land \forall j \in [2, k] : b_{j-1} \in b_j.parents$
        \EndProcedure
    \end{algorithmic}
\end{algorithm}

\subsection{The decision rules}
In \Cref{alg:decider}, by the direct decision rule, a slot is decided as skip if \pskipthreshold blocks of the vote round vote for no block of the slot, as commit if \pcommitthreshold blocks of the certify round certify one of its leader blocks, and stays undecided otherwise. By the indirect decision rule, the anchor is the first later slot, beyond the wave, that is not decided as skip. While the anchor is undecided, the slot stays undecided. Once the anchor is decided as commit, the slot is decided as commit if at least \plinksize certify-round blocks are simultaneously certificates for its leader block and ancestors of the anchor, and as skip otherwise.

For Mysticeti~\cite{mysticeti} ($\pwavelength{=}3$), a certificate is a block carrying \pcommitthreshold votes; all thresholds are $n{-}f$; one certified link suffices because after a direct skip no certificate can exist. For Blue Bottle~\cite{blue-bottle} ($\pwavelength{=}2$), the vote is the certificate, and $\plinksize = n{-}3f$ keeps a commit by the indirect decision rule consistent with a skip by the direct decision rule. For Nemo-Nemo~\cite{nemonemo} (crash faults), thresholds are $c{+}1$; the skip threshold $n$ never triggers in practice, so skips are decided by the indirect decision rule; equivocation does not arise. For Orcaella~\cite{orcaella} ($\pwavelength{=}2$, mixed faults), the vote is likewise the certificate, with thresholds of $n{-}f{-}c$ and $\plinksize = n{-}3f{-}2c$ certified links keeping a commit by the indirect decision rule consistent with a skip by the direct decision rule; setting $c{=}0$ recovers Blue Bottle exactly.

\begin{algorithm}[t]
    \caption{Decision Rules (the only part that varies across base protocols, \Cref{tab:base-protocols})}
    \label{alg:decider}
    \appendixalgsize
    \begin{algorithmic}[1]
        \State $\pcommitthreshold$ \Comment{Direct rule, commit: $n-f$ (Mysticeti, Blue Bottle); $n-c$ (Nemo-Nemo); $n-f-c$ (Orcaella)}
        \State $\pskipthreshold$ \Comment{Direct rule, skip: $n-f$ (Mysticeti, Blue Bottle); $n$ (Nemo-Nemo); $n-f-c$ (Orcaella)}
        \State $\plinksize$ \Comment{Indirect rule, commit: $1$ (Mysticeti, Nemo-Nemo); $n-3f$ (Blue Bottle); $n-3f-2c$ (Orcaella)}
        \Statex

        \Procedure{\piscert}{$b_{cert}, b_{leader}$}
        \State $w \gets$ \Call{\pwavenumber}{$b_{leader}.round$}
        \If{\Call{\pcertifyround}{$w$} $=$ \Call{\pvoteround}{$w$}} \Comment{Two-round waves: vote $=$ certificate}
        \State \Return \Call{\pisvote}{$b_{cert}, b_{leader}$}
        \EndIf
        \State $V \gets \{b \in b_{cert}.parents : \Call{\pisvote}{b, b_{leader}}\}$
        \State \Return $|V| \geq \pcommitthreshold$
        \EndProcedure
        \Statex

        \Procedure{\pskippedleader}{$w$} \Comment{Blocks voting for no block of the slot}
        \State $r_{vote} \gets$ \Call{\pvoteround}{$w$}
        \State $B_{vote} \gets \{b \in \pdag[r_{vote}] : \forall b_{leader} \in \Call{\pleaderblock}{w},\ \neg\Call{\pisvote}{b, b_{leader}}\}$
        \State \Return $|B_{vote}| \geq \pskipthreshold$
        \EndProcedure
        \Statex

        \Procedure{\psupportedleader}{$w, b_{leader}$}
        \State $B_{decision} \gets$ \Call{\pgetdecisionblocks}{$w$}
        \State \Return $|\{b \in B_{decision} : \Call{\piscert}{b, b_{leader}}\}| \geq \pcommitthreshold$
        \EndProcedure
        \Statex

        \Procedure{\piscertifiedlink}{$b_{anchor}, b_{leader}$}
        \State $w \gets$ \Call{\pwavenumber}{$b_{leader}.round$}
        \State $B_{decision} \gets$ \Call{\pgetdecisionblocks}{$w$}
        \State \Return $|\{b \in B_{decision} : \Call{\piscert}{b, b_{leader}} \land \Call{\pislink}{b, b_{anchor}}\}| \geq \plinksize$
        \EndProcedure
        \Statex

        \Procedure{\ptrydirectdecide}{$w$}
        \If{\Call{\pskippedleader}{$w$}} \Return $\texttt{Skip}(w)$
        \EndIf
        \For{$b_{leader} \in$ \Call{\pleaderblock}{$w$}} \Comment{Loop over equivocations}
        \If{\Call{\psupportedleader}{$w, b_{leader}$}} \Return $\texttt{Commit}(b_{leader})$
        \EndIf
        \EndFor
        \State \Return $\pundecided$
        \EndProcedure
        \Statex

        \Procedure{\ptryindirectdecide}{$w, S$}
        \State $r_{certify} \gets$ \Call{\pcertifyround}{$w$}
        \State $s_{anchor} \gets$ first $s \in S$ s.t. $r_{certify} < s.round \land s \neq \texttt{Skip}$
        \If{$s_{anchor} = \texttt{Commit}(b_{anchor})$}
        \If{$\exists b_{leader} \in \Call{\pleaderblock}{w} : \Call{\piscertifiedlink}{b_{anchor}, b_{leader}}$}
        \State \Return $\texttt{Commit}(b_{leader})$
        \Else\ \Return $\texttt{Skip}(w)$
        \EndIf
        \EndIf
        \State \Return $\pundecided$ \Comment{The anchor is undecided or not found}
        \EndProcedure
    \end{algorithmic}
\end{algorithm}

\section{Full Proofs}
\label{app:proofs}

\ifextended This appendix proves every statement of \Cref{sec:proofs}, in the same order and under the same names; each statement is restated before its proof.\fi{} Each proof names the assumption of \Cref{sec:overview} at the step that consumes it. The Lean development (\Cref{app:lean}) machine-checks every statement; the proofs here are the human-readable arguments, and they follow the machine-checked ones step for step.

\begin{appstatement}{fact:dag}{fact}{DAG facts}
    Let $b$ be a block held by an honest validator. (i)~The validator's DAG contains the entire causal history of $b$. (ii)~Every honest validator eventually holds $b$. (iii)~Any two honest validators that hold $b$ hold the same causal history of $b$.
\end{appstatement}

\begin{proof}
    (i) A2 states that a validator adds a block to its DAG only once the block's entire causal history is available; by induction over the order of additions, every block in the DAG has its causal history in the DAG.
    (ii) The broadcast layer guarantees that a block added by one honest validator is eventually received by every honest validator, which then adds it, its causal history being available by the same guarantee applied recursively.
    (iii) A block names its references as part of its content, so the causal history of $b$ --- the set of blocks reachable from $b$ by following references --- is determined by $b$ alone; by (i), each of the two validators holds exactly this set. \qed
\end{proof}

\begin{appstatement}{lemma:window}{lemma}{Window Agreement}
    Any two honest validators that commit the same pivot compute the same window.
\end{appstatement}

\begin{proof}
    \pgetsubdag (\Cref{alg:dual_mode}) defines the window as the restriction of the pivot's causal history to the \pinterval rounds up to the pivot's round; the round of a block is part of its content. A validator that commits the pivot holds its entire causal history, by \Cref{fact:dag}(i), so the window is fully contained in its DAG; by \Cref{fact:dag}(iii), the two validators hold the same causal history; and \pgetsubdag is deterministic, so the two restrictions are the same window. \qed
\end{proof}

\begin{appstatement}{prop:agreement}{proposition}{Leader-Count Agreement}
    For any deterministic update rule --- the AIMD rule of \Cref{alg:dual_mode} is one instance --- any two honest validators agree on every configuration both have reached, its start round and its leader count, and on the verdict of every slot of every configuration both have closed.
\end{appstatement}

\begin{proof}
    By induction over $k$. For the base case, both validators initialize $C_0 = (0, 1)$ with back-off $0$ (\Cref{alg:main,alg:dual_mode}). For the inductive step, assume both validators have reached $C_k$ and agree on its start round, leader count, and back-off; we show that they agree on every verdict of $C_k$'s rounds, and, if both close $C_k$, on the pivot and on $C_{k+1}$.

    \emph{Verdicts agree.} By the update semantics (\Cref{sec:integration}), every slot of the rounds from $C_k$'s start round up to and including its pivot's round is decided against the fixed schedule of $C_k$ --- $\ell(C_k)$ slots per round, slot $(r, l)$ led by $\pgetleader(r + l)$. This includes the derivations that the indirect decision rule performs: an indirect decision for a slot of those rounds reads decision anchors at higher rounds, and on the invocation that produces the verdict, those higher slots are evaluated under the same count, since \ptrycommit switches the count only after returning (\Cref{alg:main}). Every verdict of $C_k$ at either validator is therefore a verdict of one fixed schedule, and the safety half of A4 makes the two validators' verdicts equal wherever both have decided.

    \emph{Pivots agree.} Suppose both validators close $C_k$, with pivots $a_1$ and $a_2$, each the least committed slot whose round exceeds $C_k$'s start round by more than \pinterval. If $a_1 \neq a_2$, assume without loss of generality $a_1 < a_2$ in slot order. Slot $a_1$ is committed at the first validator, so by verdict agreement its verdict at the second validator is the same commit; but $a_1$ lies past the threshold and below $a_2$, contradicting the minimality of $a_2$. Hence $a_1 = a_2$, and by verdict agreement both validators commit the same pivot block.

    \emph{The next configuration agrees.} $C_{k+1}$'s start round is the pivot's round, agreed. Its leader count and back-off are the output of the update rule, a deterministic function of the agreed leader count and back-off of $C_k$, the agreed pivot block, and its window, agreed by \Cref{lemma:window}. Hence $C_{k+1}$ is equal at both validators.

    A validator that has not closed $C_k$ agrees, by the induction so far, on the strict prefix of configurations it has reached, which is all the statement claims. No step above uses synchrony, and no base-protocol fact beyond the safety half of A4 is consumed. \qed
\end{proof}

\begin{restatement}{thm:safety}{Safety}
    \sysname composed with any base protocol satisfying A1--A4 satisfies the safety of \Cref{def:bab}: the committed sequences of any two honest validators agree as far as both reach and grow only by appending (Agreement, Total Order), and no block is committed twice (Integrity).
\end{restatement}

\begin{proof}
    A validator's committed sequence up to its $K$-th closed configuration is the concatenation, in configuration order, of the per-configuration ledgers: for each $C_k$, the committed blocks of $C_k$'s rounds in slot order, as read from the verdicts. By \Cref{prop:agreement}, two validators agree on every configuration and every verdict of the configurations both have closed, so their per-configuration ledgers are equal up to the lower of the two heights, and the concatenations agree as far as both reach. A validator's own sequence at a lower height is by construction a prefix of its sequence at a higher one: later configurations only append. This gives Agreement and Total Order.

    For Integrity, suppose a block $L$ appears twice. A commit verdict names a block of its slot's round and leader (the decision rules only ever commit a leader block of the slot under decision), so both occurrences of $L$ lie at slots with $L$'s round and $L$'s author as leader. Within one configuration, the leaders of one round are distinct validators --- the rotation assigns $\pgetleader(r), \dots, \pgetleader(r + \ell - 1)$, which are distinct because $\pmaxproposersperround \leq n$ --- so the two slots are one slot, and a slot carries one verdict. Across configurations, the rounds of distinct configurations are disjoint, and $L$'s round lies in exactly one of them. \qed
\end{proof}

\begin{appstatement}{remark:conservativity}{remark}{Conservativity}
    With the constant update rule, \sysname is the base protocol: the same schedule, the same verdicts, the same ledger. The AIMD rule additionally keeps the leader count within $[1, \pmaxproposersperround]$, and its threshold test is one integer comparison.
\end{appstatement}

\begin{proof}
    Under the constant rule, every update returns the leader count and back-off it was given, so by induction from $C_0 = (0, 1)$ every configuration has one leader and back-off $0$. With one leader per round, slot numbers coincide with round numbers and the schedule is the base protocol's own; every verdict is then a derivation of the base protocol's decision rules on the unmodified DAG, and the committed sequence is the base protocol's ledger. The bounds of the AIMD rule are read off \Cref{alg:dual_mode}: the increase is capped at \pmaxproposersperround and the decrease is floored at $1$, and the threshold test multiplies both sides into integers. \qed
\end{proof}

\begin{appstatement}{lemma:progress}{lemma}{Configuration Progress}
    After GST, if every honest validator has closed configuration $C_k$, then some honest validator closes $C_{k+1}$ within bounded time, and every honest validator does within a further $\Delta$.
\end{appstatement}

\begin{proof}
    Closing $C_{k+1}$ means committing its pivot: the least committed slot whose round exceeds $C_{k+1}$'s start round by more than \pinterval. Fix an honest validator $v$. After GST, the liveness half of A4, applied to the fixed schedule of $C_{k+1}$, guarantees two things within bounded time: every slot of $C_{k+1}$'s rounds up to any fixed horizon is decided (with the wave-plus-commit-gap margin below the horizon), and some slot at a round within the commit gap $c$ of the threshold round is committed. Let $\kappa_0$ be such a committed slot; the pivot exists and lies at or below $\kappa_0$ in slot order.

    Two identifications, and this is where the lemma consumes base \emph{safety}: the commit that liveness supplies is a verdict on the network's joint view, while the pivot that the update reads is $v$'s own verdict; the safety half of A4 makes them the same verdict. Likewise, every slot between $C_{k+1}$'s start and $\kappa_0$ is decided, and \ptrycommit's walk (\Cref{alg:main}), which stops at the first undecided slot, therefore reaches the pivot and commits it. Validator $v$ then invokes the update rule and closes $C_{k+1}$.

    Finally, within a further $\Delta$, every honest validator holds the pivot and its causal history (\Cref{fact:dag}(ii), with delivery bounded by $\Delta$ after GST), decides the same slots by the safety half of A4, and closes $C_{k+1}$ by determinism of the update rule (\Cref{prop:agreement}). \qed
\end{proof}

\begin{restatement}{thm:liveness}{Liveness}
    After GST, every honest validator's run reaches every height: for every $k$ it closes configuration $C_k$, and committed slots keep appearing.
\end{restatement}

\begin{proof}
    By induction on $k$: height $0$ holds initially, and \Cref{lemma:progress} extends every height by one within bounded time after GST. Each closing commits at least the pivot, so committed slots keep appearing. Each pivot's round exceeds the previous start round by more than \pinterval, so configuration start rounds strictly increase even though the leader count oscillates. The per-configuration clause of A4-liveness that \Cref{lemma:progress} invokes carries a margin: a slot is guaranteed decided only when a commit gap and a wave still fit below the horizon of the DAG. The margin is necessary, not an artifact of the proof: a slot whose direct decision fails --- a Byzantine leader can always cause this --- is decided only through a committed slot at least a full wave above it, whose own wave must also fit under the horizon; no protocol decides such a slot at the last rounds under the horizon. \qed
\end{proof}

\begin{appstatement}{remark:rotation}{remark}{Liveness of the rotation}
    For the leader schedule used in this paper, the per-configuration liveness clause of A4 is proved rather than assumed, for all four base protocols at every leader count (\Cref{app:lean}), with a commit gap of $n + \pwavelength - 1$ rounds.
\end{appstatement}

\begin{proof}
    Each of the four base protocols supplies two descent facts after GST, from its own theorems: a slot led by a validator of its reliable set is decided as commit by the direct rule within its wave, whatever the schedule; and a slot is decided as soon as some committed slot lies at least a full wave above it with no eligible slot in between (the indirect decision rule). The reliable set misses at most $f$ validators for Mysticeti and Blue Bottle, at most $f + c$ for Orcaella, whose mixed bound enters only here, and at most $n - \lceil (n+1)/2 \rceil$ for Nemo-Nemo, whose crash bound enters only here, through the reliable set being a majority.

    The classic liveness argument asks for a wave of consecutive honest-led \emph{slots}, and under the rotation $\pgetleader(r + l)$ it has no instance once several leaders share a round: at $n = 4$ and two leaders per round, three consecutive rounds already name every validator. The proof instead descends through round \emph{heads}. Under the schedule with $m$ leaders, the head of round $\rho$ --- its first slot --- is slot $m\rho$, and it is led by $\pgetleader(\rho)$ \emph{whatever $m$ is}; this is what lets one argument serve every leader count. If the heads of \pwavelength consecutive rounds are all led by reliable validators, the last of them is decided as commit by the direct rule, and every slot within a wave below the first is decided: the committed head lies exactly a wave above it, and no slot strictly between is a full wave above it, so the indirect decision rule applies with nothing in between.

    It remains to find such heads: a pigeonhole on the rotation. Round-robin visits the $n$ validators cyclically, so the heads of \pwavelength consecutive rounds are \pwavelength consecutive validators in the rotation order. If every window of \pwavelength consecutive rounds in a cycle contained a head led outside the reliable set, then mapping each of the $n$ window starts to such a head would inject the $n$ starts into the \pwavelength positions times the at most $s$ excluded validators, giving $n \leq \pwavelength \cdot s$; the committee bounds ($n \geq 3f + 1$ with $\pwavelength = 3$ for Mysticeti, $n \geq 5f + 1$ with $\pwavelength = 2$ for Blue Bottle, $n \geq 5f + 3c + 1$ with $\pwavelength = 2$ and slack $f + c$ for Orcaella, and a majority with $\pwavelength = 2$ for Nemo-Nemo) exclude this. Hence within any $n + \pwavelength - 1$ consecutive rounds some \pwavelength consecutive heads are reliable-led, which bounds the commit gap by $n + \pwavelength - 1$. \qed
\end{proof}

\section{Lean Formalization}
\label{app:lean}

Every statement of \ifextended\Cref{sec:proofs}\fi{} is machine-checked in Lean~4. The theorems are stated once over an abstract base protocol and an arbitrary deterministic update rule, then instantiated at Mysticeti, Blue Bottle, Nemo-Nemo, and Orcaella. The formalization contains no \texttt{sorry} and relies on no axioms beyond Lean's three standard ones (propositional extensionality, quotient soundness, and choice). The development spans ${\sim}3{,}900$ lines, partitioned so a human reader audits only definitions and theorem statements (${\sim}1{,}900$ lines). The proofs (${\sim}2{,}000$ lines) are AI-generated, checked by the kernel and need no reading. A further ${\sim}3{,}200$ lines of executable witnesses evaluate every definition on concrete DAGs before any theorem uses it, guarding against vacuous statements. These witnesses include a run whose leader count rises at a pivot, two validators' runs identified by the agreement theorem, a DAG on which the liveness theorems produce verdicts end to end, and a mixed-fault DAG that one admissible indirect threshold commits and another skips.

\begin{table}[t]
    \caption{Paper statements (\ifextended\Cref{sec:proofs}\fi) and their machine-checked counterparts. Each Lean result is a \texttt{Statement} file (audited) with a separate kernel-checked proof.}
    \label{tab:lean}
    \centering
    \scriptsize
    \begin{tabular}{@{}lp{1.9cm}p{5.8cm}@{}}
        \hline
        Paper statement                  & Lean result                                                               & Content                                                                                                                                                             \\
        \hline
        Lemma (Window Agreement)         & \texttt{Window}                                                           & A view holding the pivot holds its whole causal history; two views compute one window.                                                                              \\
        Prop.\ (Leader-Count Agreement)  & \texttt{Agreement}                                                        & Two runs, from any views to any heights, agree on every configuration (start round, count, back-off) and every verdict of their common ranges; for any update rule. \\
        Theorem (Safety)                 & \texttt{Ledger}                                                           & The committed sequence is agreed as far as both runs reach, grows by prefixes, and holds no block twice.                                                            \\
        Remark (Conservativity)          & \texttt{Conservativity}                                                   & Under the constant rule every verdict is a verdict of the base protocol at one leader.                                                                              \\
        AIMD rule (\Cref{alg:dual_mode}) & \texttt{Aimd}                                                             & The count stays in $[1, \pmaxproposersperround]$; the healthy and unhealthy steps; the integer threshold test.                                                      \\
        Lemma (Configuration Progress)   & \texttt{Progress}                                                         & A run past the synchrony round extends by one configuration; consumes the agreement law.                                                                            \\
        Theorem (Liveness)               & \texttt{Progress}                                                         & Runs of every height exist under the horizon that height needs.                                                                                                     \\
        Remark (Rotation liveness)       & \texttt{Heads}                                                            & Round-robin is live at every leader count, from the descent laws and a pigeonhole, with gap $n + \pwavelength - 1$.                                                 \\
        Base protocols (A1--A4)          & \texttt{Mysticeti}, \texttt{Odontoceti}, \texttt{Nemo}, \texttt{Orcaella} & Each protocol satisfies the laws and the descent laws; Blue Bottle appears as Odontoceti, Nemo-Nemo's safety needs no fault bound.                                  \\
        \hline
    \end{tabular}
\end{table}

\subsection{The Base-Protocol Interface}
The base protocol is formalized as a structure providing a universe of blocks with views, a direct decision predicate reading one wave, and a decision relation parametric in the leader schedule. Its laws are the assumptions A1--A4 of \Cref{sec:overview} made precise --- A1 is data rather than a law, the laws capture A2, A3, and the safety half of A4, and the liveness half is the clause stated further down --- and they are exactly what \Cref{sec:proofs} consumes. The fault model lives only in each protocol's instantiation, never in the interface.

A base protocol exposes views $V$ of a universe of blocks, each carrying a round, an author, and references (A1); a predicate $\mathsf{DirectCommit}_V(L, r)$; and a relation $\mathsf{Decided}_S(V, s, v)$, for $S$ a leader schedule assigning a leader to every slot (A1), $s$ a slot, and $v$ a verdict (a block or skip), subject to:
\begin{itemize}
    \item \emph{Causal completeness (A2).} A view contains every block that its blocks reference.
    \item \emph{Agreement (A4, safety).} For a fixed schedule $S$: $\mathsf{Decided}_S(V_1, s, v_1)$ and $\mathsf{Decided}_S(V_2, s, v_2)$ imply $v_1 = v_2$.
    \item \emph{Direct decisions are decisions (A3).} If $L$ is a candidate of slot $s$ (right round, right author) and $\mathsf{DirectCommit}_V(L, \mathrm{round}(s))$, then $\mathsf{Decided}_S(V, s, L)$.
    \item \emph{Decisions are about candidates.} $\mathsf{Decided}_S(V, s, L)$ implies $L$ is a candidate of $s$.
\end{itemize}

\subsection{Runs and Liveness}
The \sysname mechanism relies on a dynamic schedule. The schedule of a configuration with $m$ leaders assigns slot $(r, l)$ to $\pgetleader(r + l)$. The leader of a slot does not depend on $m$; only which slots exist does. The window is the pivot's causal history. The observed count is over slots: a slot counts when it has a candidate decided as commit by the direct rule within the window. Two such candidates of one slot are one block, by agreement. The threshold is the integer pair $(96, 100)$ and the test an integer comparison. An update rule is any function mapping $(\text{count}, \text{backoff}, \text{DAG}, \text{pivot})$ to $(\text{count}', \text{backoff}')$. The AIMD rule of \Cref{alg:dual_mode} and the constant rule are two instances.

Runs are formalized as prefixes. A run closed to height $K$ carries configurations $0, \dots, K$ defining a start round, leader count, and back-off. Each configuration is closed by a pivot, which is the least committed slot past the interval. Every verdict of a configuration's rounds is derived against that configuration's fixed schedule, formalizing the semantics of \Cref{sec:integration}. The next configuration is the update rule applied to the pivot. A finite DAG closes finitely many configurations, so there is no total run. Every theorem compares runs of any two heights on the configurations both reach, matching the prefix formulation of \ifextended\Cref{sec:proofs}\fi.

Liveness avoids physical clocks by treating synchrony structurally. Each protocol defines when a DAG is good from a round to a horizon $N$, meaning a reliable quorum is synchronized and populating the rounds. Liveness is a clause on a schedule $S$ with commit gap $c$, captured by the following property:

\begin{quote}
    \emph{On a good DAG, (a) every slot lying at least $c$ rounds plus a wave under the horizon is decided, and (b) within any $c$ consecutive rounds in that region some slot is committed.}
\end{quote}

The margin under the horizon is the one \ifextended\Cref{sec:proofs}\fi{} shows necessary.

\subsection{Results}
\Cref{tab:lean} maps each statement of \ifextended\Cref{sec:proofs}\fi{} to its machine-checked counterpart. Mysticeti satisfies the interface laws and the descent laws at slack $f$ with commit gap $n + 2$. Blue Bottle satisfies them at slack $f$ with gap $n + 1$. Nemo-Nemo's safety consumes no fault bound at all. Its liveness uses the crash bound only for a synchronized majority of live validators to exist, with gap $n + 1$. Orcaella satisfies the laws at slack $f + c$ with gap $n + 1$; its mixed bound enters only through the reliable set.

\end{document}